%% file: main.tex
\documentclass[a4paper,USenglish]{llncs}
\usepackage{amsmath,amssymb}
\usepackage{tabls}
\usepackage{graphicx}
\usepackage{wrapfig}
\usepackage{array}
\usepackage[algoruled,linesnumbered,algo2e,vlined]{algorithm2e}
\usepackage{url}
\usepackage{multirow}
\usepackage{multicol}

\newcommand{\Substr}{\mathit{Substr}}
\newcommand{\Prefix}{\mathit{Prefix}}

\newcommand{\Endpos}{\mathit{EndPos}}
\newcommand{\RLE}{\mathit{RLE}}

\newcommand{\Rank}{\mathit{rank}}
\newcommand{\Select}{\mathit{select}}
\newcommand{\NumSetBits}{\mathit{pc}}
\newcommand{\Meta}[1]{{\langle#1\rangle}}

\newcommand{\PackedDAWG}{packed DAWG}
\newcommand{\Points}{\mathit{Points}}
\newcommand{\RLEPoints}{\mathit{Pts}}
\newcommand{\RLEDomPoints}{\mathit{Dom}}

\newcommand{\revert}[1]{#1^{rev}}

\newcommand{\PM}{\mathit{M}}
\newcommand{\PP}{\mathit{X}}

\newcommand{\headStr}{\mathit{\alpha}}
\newcommand{\bodyStr}{\mathit{\beta}}
\newcommand{\tailStr}{\mathit{\gamma}}

\newcommand{\maxe}{\mathrm{mpe}}
\newcommand{\longest}[1]{\overleftarrow{#1}}

\newcommand{\Rangequery}{\mathit{find\_any}}
\newcommand{\minpos}{\mathit{pos}}

\newcommand{\ignore}[1]{}

\author{
  Jun'ichi~Yamamoto
  \and
  Tomohiro~I
  \and
  Hideo~Bannai %
  \and
  Shunsuke~Inenaga %
  \and
  Masayuki~Takeda %
}
\institute{
  Department of Informatics, Kyushu University\\
  \email{\{tomohiro.i,bannai,inenaga,takeda\}@inf.kyushu-u.ac.jp}\\
}
\title{
  Faster Compact On-Line Lempel-Ziv Factorization 
}
\date{}
\begin{document}
\maketitle
\begin{abstract}
  We present a new on-line algorithm for computing the Lempel-Ziv
  factorization of a string
  that runs in $O(N\log N)$ time and uses only $O(N\log\sigma)$ bits of
  working space, where
  $N$ is the length of the string and $\sigma$ is the size of the alphabet.
  This is a notable improvement compared to the performance of previous on-line
  algorithms using the same order of working space but
  running in either $O(N\log^3N)$ time (Okanohara \& Sadakane 2009) or
  $O(N\log^2N)$ time (Starikovskaya 2012).
  The key to our new algorithm is in the utilization of
  an elegant but less popular index structure called Directed Acyclic Word
  Graphs, or DAWGs (Blumer et al. 1985).
  We also present an opportunistic variant of our algorithm,
  which, given the run length encoding of size $m$ of a string of length $N$, 
  computes the Lempel-Ziv factorization on-line, in 
  $O\left(m \cdot \min \left\{\frac{(\log\log m)(\log \log N)}{\log\log\log N}, \sqrt{ \frac{\log m}{\log \log m}} \right\}\right)$ time and
  $O(m\log N)$ bits of space,
  which is faster and more space efficient when the string is run-length compressible.
\end{abstract}
\input{introduction}
\input{preliminaries}

\input{dawgPacked}
\input{rle}

\bibliographystyle{splncs03}
\bibliography{ref}
\clearpage
\input{appendix}
\end{document}

%% file: introduction.tex
\section{Introduction}
The Lempel-Ziv (LZ) factorization of a string~\cite{LZ77},
discovered over 35 years ago,
captures important properties concerning repeated occurrences
of substrings in the string,
and has numerous applications in the field of data compression,
compressed full text indices~\cite{kreft11:_self_index_based_lz77},
and is also the key component
to various efficient algorithms on strings~\cite{kolpakov99:_findin_maxim_repet_in_word,duval04:_linear}.
Therefore, a large amount of work has been devoted to its
efficient computation, especially in the {\em off-line} setting where
the text is static,
and the LZ factorization can be computed in as fast as $O(N)$ time
assuming an integer alphabet, using $O(N\log N)$ or less bits of
space (See~\cite{a.ss:_lempel_ziv_lz77} for a survey; more
recent results are~\cite{ohlebusch11:_lempel_ziv_factor_revis,kempa13:_lempel_ziv,goto13:_simpl_faster_lempel_ziv_factor}).
There is much less work for the {\em on-line} setting, where new
characters may be appended to the end of the string.
If we may use $O(N\log N)$ bits of space,
the problem can be solved in $O(N\log \sigma)$ time
where $\sigma$ is the size of the alphabet, 
by use of string indicies such as suffix trees~\cite{Weiner}
and on-line algorithms to construct them~\cite{Ukk95}.
However, when $\sigma$ is small and $N$ is very large (e.g. DNA),
the $O(N\log N)$ bits space complexity is much larger than the
$N\log\sigma$ bits of the input text, and can be prohibitive.
To solve this problem, space efficient on-line algorithms 
for LZ factorization based on succinct data structures have been
proposed.
Okanohara and
Sadakane~\cite{okanohara08:_onlin_algor_findin_longes_previous_factor}
gave an algorithm that runs in $O(N\log^3 N)$ time using
$N\log\sigma + o(N\log\sigma) + O(N)$ bits of
space. 
Later Starikovskaya~\cite{starikovskaya12:_comput_lempel_ziv_factor_onlin},
achieved $O(N\log^2 N)$ time using $O(N\log\sigma)$ bits of space,
assuming $\log_\sigma N$ characters are packed in a machine word.

In this paper, we propose a new on-line LZ factorization algorithm
running in $O(N\log N)$ time using only $O(N\log\sigma)$ space,
which is a notable improvement compared to the run-times of the previous on-line algorithms
while still keeping the working space within a constant factor of the input text.
Our algorithm is based on a novel application of a
full text index called Directed Acyclic Word Graphs, or
DAWGs~\cite{blumer85:_small_autom_recog_subwor_text}, which,
despite its elegance, has not received as much attention as suffix trees.
To achieve a more efficient algorithm,
we exploit an interesting feature of the DAWG structure that, unlike
suffix trees, allows us to collect information concerning the
left context of strings into each state in an efficient and on-line
manner.
We further show that the DAWG allows for an opportunistic variant of
the algorithm which is more time and space efficient when 
the run length encoding (RLE) of the string is small.
Given the RLE of size $m \leq N$ of the string,
our on-line algorithm runs in 
$O\left(m \cdot \min \left\{\frac{(\log\log m)(\log \log N)}{\log\log\log N}, \sqrt{ \frac{\log m}{\log \log m}} \right\}\right)= o(m \log m)$
time using $O(m \log N)$ bits of space.
This improves on the \emph{off-line} algorithm of~\cite{eltabakh08:_sbc}
which runs in $O(m \log m)$ time using $O(m \log N)$ bits of space. 

%% file: preliminaries.tex
\section{Preliminaries}
\label{sec:preliminaries}
Let $\Sigma = \{1,\ldots,\sigma\}$
be a finite integer {\em alphabet}.
An element of $\Sigma^*$ is called a {\em string}.
The length of a string $S$ is denoted by $|S|$. 
The empty string $\varepsilon$ is the string of length 0.
Let $\Sigma^+ = \Sigma^* - \{\varepsilon\}$.
For a string $S = XYZ$, $X$, $Y$ and $Z$ are called
a \emph{prefix}, \emph{substring}, and \emph{suffix} of $S$, respectively.
The set of prefixes and substrings of $S$ are denoted by $\Prefix(S)$ and $\Substr(S)$,
respectively.
The \emph{longest common prefix} (lcp) of strings $X,Y$ %
is the longest string in $\Prefix(X) \cap \Prefix(Y)$.
The $i$-th character of a string $S$ is denoted by 
$S[i]$ for $1 \leq i \leq |S|$,
and the substring of a string $S$ that begins at position $i$ and
ends at position $j$ is denoted by $S[i..j]$ for $1 \leq i \leq j \leq |S|$.
For convenience, let $S[i..j] = \varepsilon$ if $j < i$.
A position $i$ is called an {\em occurrence} of $X$ in $S$ if $S[i..i+|X|-1]
= X$.
For any string $S = S[1..N]$, let $\revert{S} = S[N]\cdots S[1]$
denote the reversed string.
For any character $a \in \Sigma$ and integer $i\geq 0$,
let $a^0 = \varepsilon$, $a^i = a^{i-1}a$.
We call $i$ the {\em exponent} of $a^i$.

The default base of logarithms will be 2.
Our model of computation is the unit cost word RAM
with the machine word size at least $\log N$ bits.
For an input string $S$ of length $N$,
let $r = \log_\sigma N = \frac{\log N}{\log\sigma}$. 
For simplicity, assume that $\log N$ is divisible by $\log\sigma$,
and that $N$ is divisible by $r$.
A string of length $r$, called a {\em meta-character}, consists of
$\log N$ bits, and therefore fits in a single machine word.
Thus, a meta-character can also be transparently regarded as an
element in the integer alphabet
$\Sigma^r = \{1,\ldots,N\}$.
We assume that given $1\leq i \leq N-r+1$, 
any meta-character $A=S[i..i+r-1]$ can be retrieved in constant time.
Also, we can pre-compute an array of size $2^{\frac{\log N}{2}}$
occupying $O(\sqrt{N}\log N) = o(N)$ bits in $o(N)$ time, so
$\revert{A} = \revert{(A[r/2+1..r])}\revert{(A[1..r/2])}$
can be computed in constant time.
We call a string on the alphabet $\Sigma^{r}$ of meta-characters,
a {\em meta-string}.
Any string $S$ whose length is divisible by $r$
can be viewed as a meta-string $S$ of length $n = \frac{|S|}{r}$.
We write $\Meta{S}$ when we explicitly view string $S$ as a meta-string,
where $\Meta{S}[j] = S[(j-1)r+1 .. jr]$ for each $j \in [1, n]$.
Such range $[(j-1)r+1,jr]$ of positions will be called {\em meta-blocks}
and the beginning positions $(j-1)r+1$ of meta-blocks will be called 
{\em block borders}.
For clarity, the length $n$ of a meta-string $\Meta{S}$ will be denoted by $\|\Meta{S}\|$.
Meta-strings are sometimes called  {\em packed} strings.
Note that $n \log N = N\log\sigma$.

\subsection{LZ Factorization}
There are several variants of LZ factorization, and
as in most recent work, we consider the variant also called
s-factorization~\cite{crochemore84:_linear}.
  The s-factorization of a string $S$ is
  the factorization $S = f_1 \cdots f_z$ where each
  s-factor $f_i\in\Sigma^+~(i=1,\ldots,z)$
  is defined as follows:
  $f_1 = S[1]$. For $i \geq 2$:
  if $S[|f_1\cdots f_{i-1}|+1] = c \in \Sigma$ does not occur in $f_1\cdots f_{i-1}$,
  then $f_i = c$. Otherwise, $f_i$ is the longest prefix of $f_i
  \cdots f_z$ that occurs at least twice in $f_1 \cdots f_i$.
Notice that self-referencing is allowed, i.e.,
the previous occurrence of $f_i$ may overlap with itself.
Each s-factor can be represented in a constant number of words, i.e., 
either as a single character or a pair of integers representing the
position of a previous occurrence of the factor and its length.
(See Fig.~\ref{fig:s-factorization} in Appendix A. for an example.)
\subsection{Tools}

Let $B$ be a bit array of length $N$.
For any position $x$ of $B$,
let $\Rank(B,x)$ denote the number of 1's in $B[1..x]$.
For any integer $j$, let $\Select(B,j)$ denote the position of the
$j$th 1 in $B$.
For any pair of position $x,y~(x\leq y)$ of $B$,
the number of 1's in $B[x..y]$ can be expressed as
$\NumSetBits(B,x,y) = \Rank(B,y) - \Rank(B,x-1)$.
Dynamic bit arrays can be maintained
to support rank/select queries and flip operations in $O(\log N)$
time, using $N+o(N)$ bits of space
(e.g. Raman et al.~\cite{raman01:_succin_dynam_data_struc}).\\

Directed Acyclic Word Graphs (DAWG) are a variant of suffix indices,
similar to suffix trees or suffix arrays.
The DAWG of a string $S$ is the smallest partial deterministic finite automaton that accepts all
suffixes of $S$. Thus, an arbitrary string is a substring of $S$
iff it can be traversed from the source of the DAWG.
While each edge of the suffix tree corresponds to a
substring of $S$, an edge of a DAWG corresponds to a single character.
\begin{theorem}[Blumer et al.~\cite{blumer85:_small_autom_recog_subwor_text}]
  \label{theorem:orig_dawg}
  The numbers of states, edges and suffix links of the DAWG are 
  $O(|S|)$, independent of the alphabet size $\sigma$.
  The DAWG augmented with the suffix links
  can be constructed in an on-line manner in $O(|S|\log\sigma)$
  time using $O(|S| \log |S|)$ bits of space.    
\end{theorem}

We give a more formal presentation of DAWGs below.
Let $\Endpos_S(u) = \{ j \mid u = S[i..j], 1 \leq i \leq j \leq N \}$.
Define an equivalence relation on $\Substr(S)$
such that for any $u,w \in \Substr(S)$, $u \equiv_S w \iff \Endpos_S(u) = \Endpos_S(w)$, and denote
the equivalence class of $u\in\Substr(S)$ as $[u]_S$.
When clear from the context, we abbreviate the above notations
as $\Endpos$, $\equiv$ and $[u]$, respectively.
Note that for any two elements in $[u]$, one is a suffix of the
other (or vice versa).
We denote by $\longest{u}$ the longest member of $[u]$.
The states $V$ and edges $E$ of a DAWG can be characterized as
$V =  \{[u] \mid u \in \Substr(S)\}$ and
$E = \{([u],a,[ua]) \mid u, ua \in \Substr(S), u\not\equiv ua\}$.
We also define the set $G$ of labeled reversed edges, called {\em suffix links}, 
by $G = \{([au], a, [u]) \mid u, au \in \Substr(S), u = \longest{u}\}$.
An edge $([u], a, [ua])\in E$ is called a {\em primary} edge
if $|\longest{u}| + 1 = |\longest{ua}|$, 
and a {\em secondary} edge otherwise. 
We call $[ua]$ a primary (resp. secondary) child of $[u]$
if the edge is primary (resp. secondary).
(See Fig.~\ref{fig:DAWG} in Appendix for examples.)
By storing $|\longest{u}|$ at each state $[u]$,
we can determine whether an edge $([u],a,[ua])$ is primary or secondary 
in $O(1)$ time using $O(|S| \log |S|)$ bits of total space.

Whenever a state $u$ is created during the on-line
construction of the DAWG, it is possible to assign the position
$\minpos_{[u]} = \min\Endpos_S(u)$ to that state.
If state $u$ is reached by traversing the DAWG from the source with
string $p$,
this means that $p = S[\minpos_{[u]} - |p| + 1..\minpos_{[u]}]$, and thus
the first occurrence $\minpos_{[u]}-|p|+1$ of $p$ can be retrieved,
using $O(|S|\log |S|)$ bits of total space.\\

For any set $P$ of points on a 2-D plain, consider query $\Rangequery(P, I_h, I_t)$ 
which returns an \emph{arbitrary} element in $P$ that is contained in
a given orthogonal range $I_h \times I_t$
if such exists, and returns $\mathbf{nil}$ otherwise.
A simple corollary of the following result by Blelloch~\cite{blelloch08:_space}:
\begin{theorem}[Blelloch~\cite{blelloch08:_space}]
  \label{theorem:blelloch}
  The 2D dynamic orthogonal range reporting problem on $n$ elements
  can be solved using 
  $O(n\log n)$ bits of space
  so that insertions and deletions take $O(\log n)$ 
  amortized
  time and range
  reporting queries take $O(\log n + k \log n/ \log \log n)$
  time, where $k$ is the number of output elements.
\end{theorem}
is that the query $\Rangequery(P, I_h, I_t)$ can be answered in $O(\log n)$
time on a dynamic set $P$ of points.
It is also possible to extend the $\Rangequery$ query to return,
in $O(\log n)$ time, a constant number of elements contained in the range.

%% file: dawgPacked.tex
\section{On-line LZ Factorization with Packed Strings}
The problem setting and high-level structure of our algorithm follows
that of Starikovskaya~\cite{starikovskaya12:_comput_lempel_ziv_factor_onlin},
but we employ somewhat different tools. 
The goal of this section is to prove the following theorem.
\begin{theorem}\label{theorem:maintheorem}
  The s-factorization of any string $S\in\Sigma^*$ of length $N$ can be computed in an on-line manner
  in $O(N \log N)$ time and $O(N\log\sigma)$ bits of space.
\end{theorem}
By on-line, we assume that the input string $S$ is given $r$
characters at a time, and we are to compute the s-factorization of the string
$S[1..jr]$ for all $j = 1, \ldots, n$.
Since only the last factor can change for each $j$,
the whole s-factorization need not be re-calculated so
we will focus on describing how to compute each s-factor $f_i$
by extending $f_i$ while a previous occurrence exists.
We show how to maintain dynamic data structures
using $O(N\log\sigma)$ bits in $O(N\log N)$ total time that
allow us to
(1) determine whether $|f_i| < r$ in $O(1)$ time, and if
so, compute $f_i$ in $O(|f_i|\log N)$ time (Lemma~\ref{lem:shorter_case}),
(2) compute $f_i$ in $O(|f_i|\log N)$
time when $|f_i| \geq r$ (Lemma~\ref{lem:longer_case}), and
(3) retrieve a previous occurrence
of $f_i$ in $O(|f_i|\log N)$ time (Lemma~\ref{lem:previousoccurrence}).
Since $\sum_{i=1}^z |f_i| = N$, these three lemmas prove Theorem~\ref{theorem:maintheorem}.

The difference between our algorithm and that of Starikovskaya
can be summarized as follows:
For (1), we show that a dynamic succinct bit-array that supports
rank/select queries and flip operations
can be used, as opposed to a suffix trie employed
in~\cite{starikovskaya12:_comput_lempel_ziv_factor_onlin}.
This allows our algorithm to use a larger meta-character size
of $r = \log_\sigma N$ instead of $\frac{\log_\sigma N}{4}$
in~\cite{starikovskaya12:_comput_lempel_ziv_factor_onlin},
where the 1/4 factor was required
to keep the size of the suffix trie within $O(N\log\sigma)$ bits.
Hence, our algorithm can pack characters more efficiently into a word.
For (2), we show that by using a DAWG on the meta-string of length $n
= N / r$ that occupies only $O(N \log \sigma)$ bits,
we can reduce the problem of finding valid extensions of a factor
to dynamic orthogonal range reporting queries,
for which a space efficient dynamic data structure
with $O(\log n)$ time query and update exists~\cite{blelloch08:_space}.
In contrast, Starikovskaya's algorithm uses a suffix tree on the meta-string
and dynamic wavelet trees requiring $O(\log^2 n)$ time for queries and
updates, which is the bottleneck of her algorithm.
For (3), we develop an interesting technique for the case
$|f_i| < r$ which may be of independent interest.

In what follows, let $l_i = \sum_{k=1}^{i-1}{|f_k|}$.
Although our presentation assumes that $N$ is known,
this can be relaxed at the cost of a
constant factor by simply %
restarting the entire algorithm when the length of the input string
doubles.

\subsection{Algorithm for $|f_i| < r$}
\label{subsec:smallerthanr}
Consider a bit array $\PM_k[1..N]$.
For any meta-character $A \in \Sigma^r$, let
$\PM_k[A] = 1$ iff $S[l+1 .. l+r] = A$ for some $0 \leq l \leq k-r$,
i.e., $\PM_k[A]$ indicates whether $A$ occurs as a substring in $S[1..k]$.
For any short string $t$ ($|t| < r$), 
let $D_{t}$ and $U_{t}$ be, respectively, the lexicographically smallest and largest
meta-characters having $t$ as a prefix,
namely, 
the bit-representation\footnote{
  Assume that $0^{\log N}$ and $1^{\log N}$ correspond
  to meta-characters $1$ and $N$, respectively.} of $D_t$ is the concatenation of
the bit-representation of $t$ and $0^{(r-|t|)\log\sigma}$,
and 
the bit-representation of $U_t$ is the concatenation of
the bit-representation of $t$ and $1^{(r-|t|)\log\sigma}$.
These representations can be obtained from $t$ in constant time
using standard bit operations.
Then, the set of meta-characters that have $t$ as a prefix 
can be represented by the interval $tr(t) = [D_t,U_t]$.
It holds that $t$ occurs in $S[1..k-r+|t|]$ iff
some element in $\PM_k[D_{t}..U_{t}]$ is 1,
i.e. $\NumSetBits(\PM_k,D_{t},U_{t}) > 0$.
Therefore, 
we can check whether or not a string of length up to $r$ occurs
at some position $p \leq l_i$ by using $\PM_{l_i+r-1}$. 

For any $0 \leq m \leq r$, let $t_m = S[l_i+1..l_i+m]$.
We have that $|f_i| < r$ iff $\PM_{l_i+r-1}[t_r] = 0$,
which can be determined in $O(1)$ time.
Assume $|f_i| < r$,
and let
$m_i = \max\{ m \mid 0 \leq m < r, 
  \NumSetBits(\PM_{l_i+r-1},D_{t_m},U_{t_m}) > 0
  \}$,
where $m_i=0$ indicates that $S[l_i+1]$ does not occur in $S[1..l_i]$.
From the definition of s-factorization, we have that $|f_i| = \max(1, m_i)$.
Notice that $m_i$ can be computed by $O(|f_i|)$ rank queries
on $\PM_{l_i+r-1}$, due to the monotonicity of
$\NumSetBits(\PM_{l_i+r-1},D_{t_m}, U_{t_m})$ 
for increasing values of $m$.
To maintain $\PM_k$
we can use rank/select dictionaries for a dynamic bit array of length
$N$~(e.g.~\cite{raman01:_succin_dynam_data_struc}) mentioned in Section~\ref{sec:preliminaries}.
Thus we have:

\begin{lemma} \label{lem:shorter_case}
  We can maintain in $O(N\log N)$ total time,
  a dynamic data structure occupying $N + o(N)$ bits of space
  that allows whether or not $|f_i| < r$ to be determined in $O(1)$
  time, and if so, $f_i$ to be computed in $O(|f_i|\log N)$ time.
\end{lemma}

\subsection{Algorithm for $|f_i| \geq r$.}
To compute $f_i$ when $|f_i| \geq r$,
we use the DAWG for the meta-string $\Meta{S}$ which we call the
{\em {\PackedDAWG}}.
While the DAWG for $S$ requires $O(N\log N)$ bits, the {\PackedDAWG}
only requires $O(N\log\sigma)$ bits. However, the complication is that
only substrings with occurrences that start at block borders can
be traversed from the source of the {\PackedDAWG}.
In order to overcome this problem, we will augment the {\PackedDAWG}
and maintain the set
$\Points_{[u]} = \{(\revert{A}, X) \mid ([u], X, [uX]) \in E, A \longest{u}X \in \Substr(\Meta{S}) \}$
for all states $[u]$ of the {\PackedDAWG}.
A pair $(\revert{A},X) \in \Points_{[u]}$ represents that 
there exists an occurrence of $A\longest{u}X$ in $\Meta{S}$,
in other words, the longest element $\longest{u}$ corresponding to the state can be extended by $X$ and still
have an occurrence in $\Meta{S}$ immediately preceded by $A$. 

\begin{lemma} \label{lem:PointsSize}
  For meta-string $S$ and its {\PackedDAWG} $(V,E,G)$,
  the the total number of elements in
  $\Points_{[u]}$ for all states $[u] \in V$ is $O(\|\Meta{S}\|)$.
\end{lemma}
\begin{proof}
  Consider edge $([u],X,[uX])\in E$.
  If $\longest{u}X \neq \longest{uX}$, 
  i.e., the edge is secondary,
  it follows that there exists a unique meta-character 
  $A = \Meta{S}[\minpos_{[uX]}-\|\longest{u}X\|]$
  such that $A\longest{u}X\equiv_{\Meta{S}} \longest{u}X$,
  namely, any occurrence of $\longest{u}X$ is always preceded by $A$ in $\Meta{S}$.
  If $\longest{u}X = \longest{uX}$, i.e. the edge is primary,
  then,
  for each distinct meta-character $A$ preceding an occurrence of
  $\longest{u}X = \longest{uX}$ in $\Meta{S}$,
  there exists a suffix link
  $([A\longest{uX}], A, [\longest{uX}]) \in G$.
  Therefore, each point $(\revert{A},X)$ in $\Points_{[u]}$
  can be associated to either a secondary edge from $[u]$
  or one of the incoming suffix links to its primary child $[uX]$.
  (See also Fig.~\ref{fig:points} in Appendix A.)
  Since each state has a unique longest member,
  each state has exactly one incoming primary edge.
  Therefore, the total number of elements in $\Points_{[u]}$ for \emph{all} states $[u]$ 
  is equal to the total number of secondary edges and suffix links,
  which is $O(\|\Meta{S}\|)$.
  \qed
\end{proof}

\begin{lemma} \label{lem:Points}
    For string $S\in\Sigma^*$ of length $N$,
    we can, in $O(N\log \sigma)$ total time and bits of space and in an on-line manner,
    construct the {\PackedDAWG} $(V,E,G)$ of $S$ as well as maintain
    $\Points_{[u]}$ for all states $[u]\in V$ so that
    $\Rangequery(\Points_{[u]},$ $I_h, I_t)$ for an orthogonal
    range $I_v\times I_h$ can be answered in $O(\log n)$ time.
\end{lemma}
\begin{proof}
  It follows from Theorem~\ref{theorem:orig_dawg}
  that the {\PackedDAWG} can be computed in an on-line manner,
  in $O(N \log \sigma)$ time and bits of space,
  since the size of the alphabet for meta-strings is $O(N)$ and the
  length of the meta-string is $n = \frac{N}{r}$.
  To maintain and support $\Rangequery$ queries on $\Points$
  efficiently, we use the dynamic data structure by
  Blelloch~\cite{blelloch08:_space} mentioned in
  Theorem~\ref{theorem:blelloch}.
  Thus from Lemma~\ref{lem:PointsSize}, the total space requirement is
  $O(N\log\sigma)$ bits.
  Since each insert operation can be performed in amortized
  $O(\log n)$ time (no elements are deleted in our algorithm),
  what remains is to show that the total number of insert operations
  to $\Points$ is $O(n)$.
  This is shown below by a careful analysis of the on-line DAWG construction
  algorithm~\cite{blumer85:_small_autom_recog_subwor_text}.
  (See Algorithm~\ref{algorithm:onlineDAWGconstruction} in Appendix
  B. for pseudo-code.)

  Assume we have the {\PackedDAWG} for a prefix $u = \Meta{S}[1..\|u\|]$
  of meta-string $\Meta{S}$.
  Let $B = \Meta{S}[\|u\|+1]$ be the meta-character that follows $u$ in $\Meta{S}$. 
  We group the updates performed on the {\PackedDAWG} when adding $B$,
  into the following two operations: %
  (a) the new sink state $[uB]$ is created, and
  (b) a state is split.

 First, consider case (a).
 Let $u_0 = u$, and 
 consider the sequence $[u_1], \ldots, [u_q]$ of states
 such that the suffix link of $[u_{j}]$ points to $[u_{j+1}]$ for $0 \leq j < q$,
 and $[u_q]$ is the first state in the sequence which has an out-going edge
 labeled by $B$.
 Note that any element of $[u_{j+1}]$ is a suffix of any element of $[u_{j}]$.
 The following operations are performed. (See also
 Fig.~\ref{fig:secondary_edges} in Appendix A.)
 (a-1)
 The primary edge from the old sink $[u]$ to the new sink $[uB]$ is created.
 No insertion is required for this edge since $[uB]$
 has no incoming suffix links.
 (a-2)
 For each $1 \leq j < q$
 a secondary edge $([u_j], B,[uB])$ is created,
 and the pair $(\revert{C_j}, B)$ is inserted to $\Points_{[u_j]}$,
 where $C_j$ is the unique meta-character that immediately precedes 
 $\longest{u_j}B$ in $uB$,
 i.e., $C_j = \Meta{uB}[pos_{[uB]}-\|\longest{u_{j}}B\|]$.
 (a-3)
 Let $([u_q],B,w)$ be the edge with label $B$ from state $[u_q]$.    
 The suffix link of the new sink state $[uB]$ is created
 and points to $w$.
 Let $e = ([v], B, w)$ be the primary incoming edge to $w$,
 and $A$ be the meta-character that labels the suffix link
 (note that $[v]$ is not necessarily equal to $[u_q]$).
 We then insert a new pair $(\revert{A}, B)$ %
 into $\Points_{[v]}$. 

Next, consider case (b).
After performing (a), node $w$ is split if the edge $([u_q],B,w)$ is secondary.
Let $[v_1] = [v]$, and 
let $[v_1], \ldots, [v_k]$ be the parents of the state $w$ of 
the {\PackedDAWG} for $u$,
sorted in decreasing order of their longest member.
Then, it holds that there is a suffix link from $[v_{h}]$ to $[v_{h+1}]$
and any element of $[v_{h+1}]$ is a suffix of any element of $[v_{h}]$
for any $1 \leq h < k$.
Assume $\longest{v_i}B$ is the longest suffix of $uB$
that has another (previous) occurrence in $uB$.
(Namely, $[v_i]$ is equal to the state $[u_q]$ of (a-2) above.)
If $i > 1$, then the state $w$ is split into two states 
$[v_1 B]$ and $[v_{i} B]$ such that $[v_1 B] \cup [v_i B] = w$ and
any element of $[v_i B]$ is a proper suffix of any element of $[v_1 B]$.
The following operations are performed.
(See also Fig.~\ref{fig:state_separation} in Appendix A.)
(b-1)
    The secondary edge from $[v_i]$ to $w$ becomes the primary edge to $[v_i B]$,
    and for all $i < j \leq k$ 
    the secondary edge from $[v_j]$ to $w$ becomes a secondary edge to $[v_j B]$.
    The primary and secondary edges from $[v_h]$ to $w$ for all $1 \leq h < i$ 
    become the primary and secondary ones from $[v_h]$ to $[v_1 B]$, respectively.
    Clearly the sets $\Points_{[v_h]}$ for all $1 \leq h < i$ are unchanged.
    Also, since any edge $([v_j], B, [v_i B])$ are all secondary,
    the sets $\Points_{[v_j]}$ for all $i < j \leq k$ are unchanged.
    Moreover, the element of $\Points_{[v_i]}$ that was associated to
    the secondary edge to $w$, is now associated to the suffix link from $[v_1 B]$ to $[v_i B]$.
    Hence, $\Points_{[v_i]}$ is also unchanged.
    Consequently, there are no updates due to edge redirection.
(b-2)
    All outgoing edges of $[v_1 B]$ are
    copied as outgoing edges of $[v_i B]$.
    Since any element of $[v_i B]$ is a suffix of any element of $[v_1 B]$,
    the copied edges are all secondary.
    Hence, we insert a pair to $\Points_{[v_i B]}$ for each secondary edge,
    accordingly.

  Thus, the total number of insert operations to $\Points$
  for all states is linear in the number of update operations during the on-line construction
  of the {\PackedDAWG}, which is $O(n)$ due to~\cite{blumer85:_small_autom_recog_subwor_text}.
  This completes the proof.
 \qed
\end{proof}

For any string $f$ and integer $0 \leq m \leq \min(|f|,r-1)$,
let strings $\headStr_m(f)$, $\bodyStr_m(f)$, $\tailStr_m(f)$ satisfy
$f = \headStr_m(f) \bodyStr_m(f) \tailStr_m(f)$,
$|\headStr_m(f)| = m$,
and $|\bodyStr_m(f)| = j^\prime r$ where 
$j^\prime = \max\{ j\geq 0 \mid m + jr \leq |f|\}$.
We say that an occurrence of $f$ in $S$ has offset $m$~($0 \leq m \leq r-1$), 
if, in the occurrence,
$\headStr_m(f)$ corresponds to a suffix of a meta-block,
$\bodyStr_m(f)$ corresponds to a sequence of meta-blocks
(i.e. $\bodyStr_m(f) \in \Substr(\Meta{S})$), and
$\tailStr_m(f)$ corresponds to a prefix of a meta-block.

Let $f_i^m$ denote the longest prefix of $S[l_i+1..N]$ 
which has a previous occurrence in $S$ with offset $m$.
Thus, $|f_i| = \max_{0 \leq m < r}|f_i^m|$.
In order to compute $f_i^m$, the idea is to 
find the longest prefix $u$ of 
meta-string $\Meta{\bodyStr_m(S[l_i+1..N])}$ 
that can be
traversed from the source of the {\PackedDAWG} while assuring that at
least one occurrence of $u$ in $\Meta{S}$ is immediately preceded by
a meta-block that has $\headStr_m(S[l_i+1..N])$ as a suffix.
It follows that $u = \bodyStr_m(f_i^m)$.

\begin{lemma}\label{lem:longestprefix}
  Given the augmented {\PackedDAWG} $(V,E,G)$ of Lemma~\ref{lem:Points} of meta-string $\Meta{S}$,
  the longest prefix $f$ of any string $P$ that has an occurrence with
  offset $m$ in $S$ can be computed in $O(\frac{|f|}{r}\log n + r \log n)$ time.
\end{lemma}
\begin{proof}
We first traverse the {\PackedDAWG} for $\Meta{S}$ to find
$\bodyStr_m(f)$. This traversal is trivial for $m=0$, 
so we assume $m > 0$.
For any string $t$ ($|t| < r$),
let $L_t$ and $R_t$ be, respectively, the lexicographically smallest and largest 
meta-character which has $t$ as a suffix,
namely,
the bit-representation of $L_t$ is the concatenation of
$0^{(r-|t|)\log\sigma}$ and the bit-representation of $t$,
and 
the bit-representation of $R_t$ is the concatenation of
$1^{(r-|t|)\log\sigma}$ and the bit-representation of $t$.
Then, the set of meta-characters that have $\revert{t}$ as a prefix,
(or, $t$ as a suffix when reversed),
can be represented by the interval $hr(t) =
[\revert{L_t},\revert{R_t}]$.
Suppose we have successfully traversed the {\PackedDAWG}
with the prefix $u = \Meta{\bodyStr_m(P)}[1..\|u\|]$ and want to
traverse with the next meta-character $X = \Meta{\bodyStr_m(P)}[\|u\|+1]$.
If $u = \longest{u}$, i.e. only primary edges were traversed,
then there exists an occurrence of 
$\headStr_m(P)u X$ with offset $m$ in string $S$
iff
$\Rangequery(\Points_{[u]}, hr(\headStr_m(P)), [X,X]) \neq \mathbf{nil}$.
Otherwise, if $u \neq \longest{u}$,
all occurrences of $u$ (and thus all extensions of $u$ that can be
traversed)
in $\Meta{S}$ is already guaranteed to be
immediately preceded by the unique meta-character
$A = \Meta{S}[\minpos_{[u]}-\|u\|]$ such that
$\revert{A} \in hr(\headStr_m(P))$.
Thus, there exists an occurrence of
$\headStr_m(P)u X$ with offset $m$ in string $S$ iff $([u], X, [uX])\in E$.
We extend $u$ until $\Rangequery$ returns $\mathbf{nil}$ or no edge is
found,
at which point we have $\headStr_m(P)u = \headStr_m(f)\bodyStr_m(f)$.

Now, $\tailStr_m(f)$ is a prefix of
meta-character $B = \Meta{\bodyStr_m(P)}[\|\Meta{u}\|+1]$.
When $u = \longest{u}$,
we can compute $\tailStr_m(f)$ by asking
$\Rangequery(\Points_{[u]}, hr(\headStr_m(P)), tr(B[1..j]))$ for
$0\leq j < r$.
The maximum $j$ such that $\Rangequery$ does not return $\mathbf{nil}$ gives
$|\tailStr_m(f)|$.
If $u\neq\longest{u}$, $\tailStr_m(f)$ is the longest lcp between
$B$ and any outgoing edge from $[u]$.
This can be computed in $O(\log n + |\tailStr_m(f)|)$ time
by maintaining outgoing edges from $[u]$ in balanced binary search trees,
and finding the lexicographic predecessor/successor $B^-,B^+$ of
$B$ in these edges, and computing the lcp between them.
(See Fig.~\ref{fig:points} in Appendix.)
The lemma follows since each $\Rangequery$ query takes $O(\log n)$ time.
\qed
\end{proof}

From the proof of Lemma~\ref{lem:longestprefix},
$\bodyStr_m(f_i^m)$ can be computed in
$O(\frac{|f_i^m|}{r}\log n)$ time, 
and for all $0 \leq m  < r$, this becomes
$O(|f_i|\log n)$ time.
However, for computing $\tailStr_m(f_i^m)$,
if we simply apply the algorithm and use $O(r\log n)$ time for each $f_i^m$,
the total time for all $0\leq m < r$ would be
$O(r^2 \log n)$ which is too large for our goal.
Below, we show that all $\tailStr_m(f_i^m)$ are not required
for computing $\max_{0\leq m < r}|f_i^m|$,
and this time complexity can be reduced.

Consider computing $F_m = \max_{0\leq x \leq m}|f_i^x|$
for $m=0,\ldots, r-1$. We first compute
$\hat{f}_i^{m} = \headStr_m(f_i^m)\bodyStr_m(f_i^m)$ using
the first part of the proof of Lemma~\ref{lem:longestprefix}. 
We shall compute $\tailStr_m(f_i^m)$ only when
$F_m$ can be larger than $F_{m-1}$ i.e.,
$|\hat{f}_i^{m}| + |\tailStr_m(f_i^m)| > F_{m-1}$.
Since $|\tailStr_m(f_i^m)| < r$, 
this will never be the case if $|\hat{f}_i^{m}| \leq F_{m-1} - r + 1$,
and will always be the case if $|\hat{f}_i^{m}| > F_{m-1}$.
For the remaining case, i.e. $0 \leq F_{m-1} - |f_i^m| < r-1$,
$F_m > F_{m-1}$
iff $|\tailStr_m(f_i^m)| >  F_{m-1} - |\hat{f}_i^{m}|$.
If $u=\longest{u}$,
this can be determined by a single $\Rangequery$ query
with $j = F_{m-1} - |\hat{f}_i^{m}|+1$ in the last part of the proof
of Lemma~\ref{lem:longestprefix},
and if so, the rest of $\tailStr_m(f_i^{m})$ is computed using
the $\Rangequery$ query for increasing $j$.
When $u\neq\longest{u}$, whether or not the lcp between
$B$ and $B^-$ or $B^+$ is greater than $F_{m-1}-|\hat{f}_i^m|$ can be
checked in constant time using bit operations.

From the above discussion, each $\Rangequery$ or
predecessor/successor query for computing
$\tailStr_m(f_i^m)$ updates $F_m$,
or returns $\mathbf{nil}$.
Therefore, the total time for computing $F_{r-1} = |f_i|$ is
$O((r+|f_i|)\log n) = O(|f_i|\log n)$.

A technicality we have not mentioned yet, is when and to what
extent the {\PackedDAWG} is updated when computing $f_i$.
Let $F$ be the length of the current longest prefix of $S[l_i+1..N]$
with an occurrence less than $l_i+1$, found so
far while computing $f_i$.
A self-referencing occurrence of $S[l_i+1..l_i+F]$ can reach up to position $l_i+F-1$.
When computing $f_i$ using the {\PackedDAWG}, $F$ is increased by at most $r$ characters at a time.
Thus, for our algorithm to successfully detect such
self-referencing occurrences,
the {\PackedDAWG} should be built up to the meta-block that includes
position $l_i+F-1+r$ and updated when $F$ increases.
This causes a slight problem when computing $f_i^m$ for some $m$;
we may detect a substring which only has an occurrence larger than $l_i$
during the traversal of the DAWG.
However, from the following lemma,
the number of such {\em future} occurrences that update $F$ can be limited 
to a constant number, namely two, and hence by reporting up to three elements in
each $\Rangequery$ query that may update $F$,
we can obtain an occurrence less than $l_i+1$, if one exists.
These occurrences can be retrieved in $O(\log N)$ time in
this case, as described in Section~\ref{subsec:retrieving_occurrences}.
\begin{lemma} \label{lem:Factor_subsequent}
During the computation of $f_i^m$, there can be at most two future
occurrences of $f_i^m$ that will update $F$.
\end{lemma}
\begin{proof}
As mentioned above, the {\PackedDAWG} is built up to the 
meta string $\Meta{S[1..s]}$
where $s = \lceil \frac{l_i + F+r-1}{r}\rceil r$.
An occurrence of $f_i^m$ possibly greater than $l_i$ 
can be written as 
$p_{m,k} = \lceil \frac{l_i}{r}\rceil r - m + 1 + kr$, where
$k = 0,1,\ldots$. %
For the occurrence to be able to update $F$ and also be detected in the 
packed DAWG, 
it must hold that
$s > p_{m,k} + F$.
Since $l_i + F + 2r -2 \geq s > p_{m,k} + F \geq l_i - m + 1 + kr + F$,
$k$ should satisfy $(2-k)r \geq 1-m$, and thus can only be $0$ or $1$.\qed
\end{proof}

\begin{lemma}\label{lem:longer_case}
  We can maintain in a total of $O(N\log N)$ time,
  a dynamic data structure occupying $O(N\log\sigma)$
  bits of space that allows
  $f_i$ to be computed in $O(|f_i|\log N)$ time, when $|f_i| \geq r$.
\end{lemma}

\subsection{Retrieving a Previous Occurrence of $f_i$}
\label{subsec:retrieving_occurrences}
If $|f_i| \geq r$, let $f_i = f_i^m$,
$\revert{A} \in hr(\headStr_m(f_i))$,
$u = \bodyStr_m(f_i)$, and $X\in tr(\tailStr_m(f_i))$ where
$A$ and $X$ were found during the traversal of the packed DAWG.
We can obtain the occurrence of $f_i$ by simple arithmetic
on the ending positions stored at each state, i.e.,
from $\minpos_{[uX]}$ if $uX\neq\longest{uX}$ or $m = 0$,
from $\minpos_{[AuX]}$ otherwise.
State $[AuX]$ can be reached in $O(\log N)$ time
from state $[uX]$, by traversing the suffix link in the reverse direction.

For $|f_i| < r$, $f_i$ is a substring of a meta-character.
Let $A_i$ be one of the previously occurring meta-characters with
prefix $f_i$ for which
$\PM_{l_i+r-1}[A_i] = 1$, thus giving a previous occurrence of $f_i$.
$A_i$ can be any meta-character in the range
$tr(f_i) = [D_{t_m}, U_{t_m}]$ with
a set bit, so $A_i$ can be retrieved in $O(\log N)$ time by
$A_i = \Select(\PM_{l_i+r-1},\Rank(\PM_{l_i+r-1},U_{t_m}))$.
Unfortunately, we cannot afford to explicitly maintain 
previous occurrences for all $N$
meta-characters, since this would cost $O(N\log N)$ bits of space.
We solve this problem in two steps.

First, consider the case that a previous occurrence of $f_i$ crosses a block border,
i.e. 
has an occurrence with some offset $1 \leq m \leq |f_i|-1$, 
and $f_i = \headStr_m(f_i)\tailStr_m(f_i)$.
For each $m = 1,\ldots, |f_i|-1$, we ask
$\Rangequery(\Points_{[\varepsilon]},$ $hr(\headStr_m(f_i)),$ $tr(\tailStr_m(f_i)))$.
If a pair $(\revert{A},X)$ is returned, this means that
$AX$ occurs in $\Meta{S}$ and
$A[r-m+1..r] = \headStr_m(f_i)$ and $X[1..\tailStr_m(f_i)] = \tailStr_m(f_i)$.
Thus, a previous occurrence of $f_i$ 
can be computed from $\minpos_{[AX]}$.
The total time required is $O(|f_i|\log n)$.
If all the $\Rangequery$ queries returned $\mathbf{nil}$, this implies that
no occurrence of $f_i$ crosses a block border and 
$f_i$ occurs only inside meta-blocks.
We develop an interesting technique to deal with this case.

\begin{lemma}
  \label{lemma:getmetablock}
  For string $S[1..k]$ and increasing values of $1\leq k \leq N$,
  we can maintain a data structure
  in $O(N\log N)$ total time and $O(N\log\sigma)$ bits of space
  that,
  given any meta-character $A$,
  allows us to retrieve a meta-character $A^\prime$ that corresponds
  to a meta block of $S$, and some integer $d$ such that
  $A^\prime[1+d..r] = A[1..r-d]$ and $0 \leq d \leq d_{A,k}$,
  in $O(\log N)$ time, where
  $d_{A,k} = \min \{ (l-1) \bmod r \mid 1 \leq l \leq k - r+1, A = S[l..l+r-1] \}$.
  (Also see Fig.~\ref{fig:insideMetaBlock} in Appendix A.)
\end{lemma}
\begin{proof}
Consider a tree $T_k$ where nodes are the set of meta-characters
occurring in $S[1..k]$. The root is $\Meta{S}[1]$.
For any meta-character $A \neq \Meta{S}[1]$,
the parent $B$ of $A$ must satisfy $B[2..r] = A[1..r-1]$ and $A\neq B$.
Given $A$, its parent $B$ can be encoded
by a single character $B[1] \in\Sigma$ that occupies $\log\sigma$ bits
and can be recovered from $B[1]$ and $A$ in constant time by simple
bit operations.
Thus, together with $\PM_k$ used in Section~\ref{subsec:smallerthanr}
which indicates which meta-characters are nodes of $T_k$,
the tree can be encoded in $O(N\log\sigma)$ bits.
We also maintain another bit vector $\PP_k$ of length $N$ so that
we can determine in constant time, whether a node in $T_k$ corresponds
to a meta-block.
The lemma can be shown if
we can maintain the tree for increasing $k$ so that for any node $A$
in the tree, either $A$ corresponds to a meta-block ($d_{A,k} = 0$), or,
$A$ has at least one ancestor at most $d_{A,k}$ nodes above it
that corresponds to a meta-block.
Assume that we have $T_{k-1}$, and want to update it to $T_k$.
Let $A = S[k-r+1..k]$.
If $A$ previously corresponded to
or the new occurrence corresponds to a meta-block,
then, $d_{A,k} = 0$ and we simply set $X_k[A] = 1$
and we are done.
Otherwise, let $B = S[k-r..k-1]$ and denote by $C$ the parent of $A$
in $T_{k-1}$, if there was a previous occurrence of $A$.
Based on the assumption on $T_{k-1}$, 
let $x_B\leq d_{B,k-1}=d_{B,k}$ and $x_C$
be the distance to the closest ancestor of $B$ and $C$, respectively,
that correspond to a meta-block. We also have that $d_{A,k-1} \geq x_C + 1$.
If $(k-r)\bmod r \geq x_C + 1$, then
$d_{A,k} = \min \{ (k-r)\bmod r, d_{A,k-1}\} \geq x_C + 1$,
i.e., the constraint is already satisfied and nothing needs to be done.
If $(k-r)\bmod r < x_C + 1$ or there was no previous occurrence of $A$,
we have that $d_{A,k} = (k-r)\bmod r$.
Notice that in such cases, we cannot have $A=B$ since that would imply
$d_{A,k} = d_{A,k-1} \neq (k-r) \bmod r$,
and thus by setting the parent of $A$ to $B$,
we have that there exists an ancestor corresponding to a meta-block at distance
$x_B + 1 \leq d_{B,k} + 1 \leq (k-r-1)\bmod r + 1 = d_{A,k}$.

Thus, what remains to be shown is how to compute $x_C$ in order to
determine whether $(k-r)\bmod r < x_C + 1$.
Explicitly maintaining the distances
to the closest ancestor corresponding to a meta-block
for all $N$ meta characters will take too much space ($O(N\log\log N)$ bits).
Instead, since the parent of a given meta-character can be obtained 
in constant time, we calculate $x_C$
by simply going up the tree from $C$, which takes $O(x_C) = O(\log N)$ time.
Thus, the update for each $k$ can be done in $O(\log N)$ time, proving the lemma.
\qed
\end{proof}
Using Lemma~\ref{lemma:getmetablock}, we can retrieve
a meta-character $A^\prime$ that corresponds to a meta-block
and an integer $0 \leq d \leq d_{A_i,k}$ such that
$A^\prime[1+d..r] = A_i[1..r-d],$ in $O(\log N)$ time.
Although $A^\prime$ may not actually occur $d$ positions prior to 
an occurrence of $A_i$ in $S[1..k]$,
$f_i$ is guaranteed to be completely contained in $A^\prime$
since it overlaps with $A_i$, 
at least as much as any meta-block actually occurring prior to
$A_i$ in $S[1..k]$.
Thus, %
$f_i = A_i[1..|f_i|] = A^\prime[1+d..d+|f_i|]$,
and $(\minpos_{[A^\prime]} -1) r + 1 + d$ is a previous occurrence of
$f_i$.
The following lemma summarizes this section.
\begin{lemma}\label{lem:previousoccurrence}
  We can maintain in $O(N\log N)$ total time, a dynamic data structure
  occupying $O(N\log\sigma)$ bits of space that allows
  a previous occurrence of $f_i$ to be computed in $O(|f_i|\log N)$
  time.
\end{lemma}

%% file: rle.tex
\section{On-line LZ factorization based on RLE} \label{sec:RLE}

For any string $S$ of length $N$, 
let $\RLE(S) = a_1^{p_1} a_2^{p_2} \cdots a_m^{p_m}$ denote the 
\emph{run length encoding} of $S$.
Each $a_k^{p_k}$ is called an RL factor of $S$,
where $a_{k} \neq a_{k+1}$ for any $1 \leq k < m$,
$p_h \geq 1$ for any $1 \leq h \leq m$,
and therefore $m \leq N$.
Each RL factor can be represented 
as a pair $(a_k,p_k) \in \Sigma\times [1..N]$,
using $O(\log N)$ bits of space.
As in the case with packed strings, we consider the on-line LZ
factorization problem, where the string is given as a 
sequence of RL factors and we are to compute the 
s-factorization of $S_j = a_1^{p_1} \cdots a_j^{p_j}$ for all $j = 1, \ldots, m$.
Similar to the case of packed strings,
we construct the DAWG of $\RLE(S)$ of length $m$,
which we will call the RLE-DAWG, in an on-line manner.
The RLE-DAWG has $O(m)$ states and edges and each edge label 
is an RL factor $a_k^{p_k}$, occupying a total of $O(m \log N)$ bits of space.
We can show that the first RL-factor of $f_i$ (corresponding to
the offset in the case of packed string), can be determined very easily,
and therefore greatly simplifies the algorithm.
Moreover, we can show that the problem of finding valid extensions of the
s-factor can be reduced to the simpler dynamic predecessor/successor
problem,
and by using the linear-space dynamic predecessor/successor data structure 
of~\cite{BeameF02}, we obtain the following result. (See Appendix for
full proof.)
\begin{theorem} \label{theo:RLE}
Given an $\RLE(S) = a_1^{p_1} a_2^{p_2} \cdots a_m^{p_m}$ of size $m$ 
of a string $S$ of length $N$,
we can compute in an on-line manner the s-factorization of $S$ in 
$O\left(m \cdot \min \left\{\frac{(\log\log m)(\log \log N)}{\log\log\log N}, \sqrt{ \frac{\log m}{\log \log m}}\right\}\right)$ 
time using $O(m \log N)$ bits of space.
\end{theorem}

%% file: appendix.tex
\appendix
\section*{Appendix A: Figures}
\begin{figure}[ph]
\centerline{
  \includegraphics[width=0.8\textwidth]{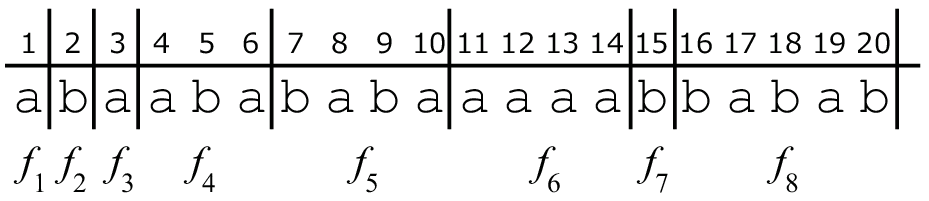}
}
\caption{
The s-factorization of the string
$S = \mathtt{abaabababaaaaabbabab}$
is
  $\mathtt{a}$,
  $\mathtt{b}$,
  $\mathtt{a}$,
  $\mathtt{aba}$,
  $\mathtt{baba}$,
  $\mathtt{aaaa}$,
  $\mathtt{b}$,
$\mathtt{babab}$. 
Each factor can be represented as a single character or a pair of
integers representing the position of a previous occurrence of the
factor and its length, i.e.,
  $\mathtt{a}$, $\mathtt{b}$, $(1,1)$, $(1,3)$,
  $(5,4)$, $(10,4)$, $(2,1)$, $(5,5)$.
Notice that the previous occurrence of an s-factor may overlap with
itself. For example, in the case of $f_5 = \mathtt{baba}$,
the two occurrences of $f_5$ are at positions $5$ and $7$.}
\label{fig:s-factorization}
\end{figure}

\begin{figure}[ph]
\centerline{
\includegraphics[width=0.8\textwidth]{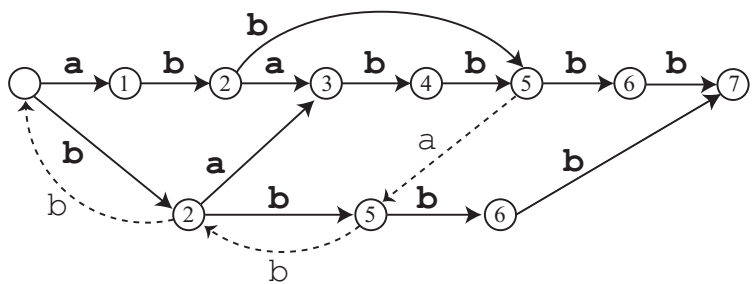}
}
\caption{The DAWG for string $\mathtt{ababbbb}$.
The solid and dashed arcs represent edges and suffix links, respectively.
For simplicity, only a subset of suffix links is shown.
The number in each state $[u]$ is $\minpos_{[u]}$.
For example, consider the non-empty suffixes of substring $\mathtt{ababb}$.
They are represented by 3 different states 
$[\mathtt{ababb}] = \{\mathtt{ababb}, \mathtt{babb}, \mathtt{abb}\}$, 
$[\mathtt{bb}] = \{\mathtt{bb}\}$, and $[\mathtt{b}] = \{\mathtt{b}\}$,
because $\Endpos_{[\mathtt{ababb}]} = \{5\}$,
$\Endpos_{[\mathtt{bb}]} = \{5,6,7\}$,
and $\Endpos_{[\mathtt{b}]} = \{2, 4, 5, 6, 7\}$.
The suffix link of $[\mathtt{ababb}]$ is labeled by character $\mathtt{a}$
and points to $[\mathtt{bb}]$.
This is because $\mathtt{bb}$ is the longest member of $[\mathtt{bb}]$
and $\mathtt{abb}$ is a member of $[\mathtt{ababb}]$.
The edge $([\mathtt{abab}], \mathtt{b}, [\mathtt{ababb}])$ is primary,
while the edge $([\mathtt{ab}], \mathtt{b}, [\mathtt{ababb}])$ is secondary.
}
\label{fig:DAWG}
\end{figure}

\begin{figure}[ph]
\centerline{
\includegraphics[width=1.0\textwidth]{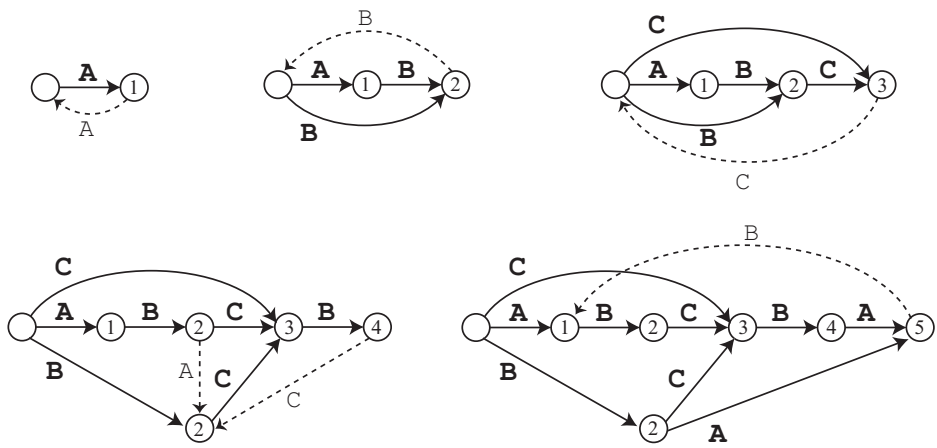}
}
\caption{A snapshot of the on-line construction of the {\PackedDAWG} of
meta string $\Meta{S} = \mathtt{ABCBA}$ for string $S = \mathtt{aaa aab aba aab aaa}$, 
where the block size $r = 3$, $\mathtt{A = aaa}$, $\mathtt{B = aab}$, and $\mathtt{C = aba}$.
The solid and dashed arcs represent edges and suffix links, respectively.
For simplicity, only the new suffix links are shown at each step.
The number in each state $[u]$ is $\minpos_{[u]}$ in $\Meta{S}$.
}
\label{fig:online_DAWG}
\end{figure}

\begin{figure}[ph]
  \centerline{
    \includegraphics[width=1.0\textwidth]{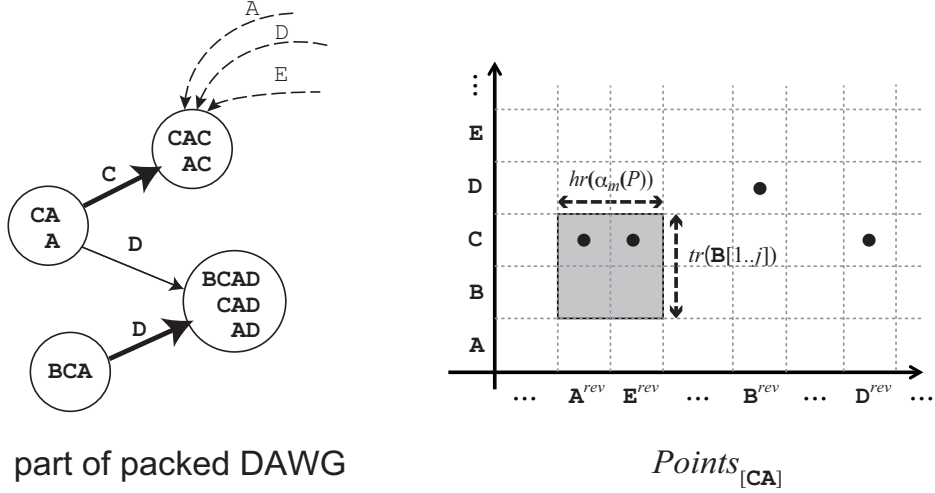}
  }
  \caption{
    Using the augmented {\PackedDAWG} to compute the longest prefix $f$
    of $P$ that occurs with offset $m$ in $S$
    (Lemma~\ref{lem:longestprefix}).
    The left figure is a part of the {\PackedDAWG} for meta-string $\Meta{S}$,
    where
    bold solid arcs represent primary edges,
    regular solid arcs represent secondary edges,
    and the dashed arcs represent suffix links.
    The right figure is the set $\Points_{[\mathtt{CA}]}$ 
    for state $[\mathtt{CA}]$.
    The pairs $(\revert{\mathtt{A}}, \mathtt{C})$,
    $(\revert{\mathtt{D}}, \mathtt{C})$,
    and $(\revert{\mathtt{E}}, \mathtt{C})$
    correspond to the incoming suffix links of state $[\mathtt{CAC}]$,
    which is a primary child of $[\mathtt{CA}]$ with the incoming primary edge labeled by $\mathtt{C}$.
    On the other hand, the pair $(\revert{\mathtt{B}}, \mathtt{D})$ corresponds to 
    the outgoing secondary edge of $[\mathtt{CA}]$ labeled by $\mathtt{D}$.
    Assume that for some $0 \leq m < r$ and string
    $P = \headStr_m(P)\mathtt{CAB}$, where
    $hr(\headStr_m(P)) = [\revert{\mathtt{A}}, \revert{\mathtt{E}}]$.
    Assume also that we have traversed the {\PackedDAWG} with $\bodyStr_m(P)[1..2] = \mathtt{CA}$,
    and want to traverse with the next meta-character $\mathtt{B}$.
    Since $\mathtt{CA} = \protect\longest{\mathtt{CA}}$ and
    there is no point in range $hr(\headStr_m(P)) \times [\mathtt{B}, \mathtt{B}]$,
    we see that $\headStr_m(P)\mathtt{CAB}$ does not occur with offset
    $m$ in $S$, so $\bodyStr_m(f) = \mathtt{CA}$.
    We then compute $\tailStr_m(f)$ %
    which is a prefix of $\mathtt{B}$,
    by querying a point in range
    $hr(\headStr_m(P)) \times tr(\mathtt{B}[1..j])]$
    for all $0 \leq j < |\tailStr_m(f)|$.
    Next, consider what happens for another string $P^\prime =
    \headStr_m(P^\prime)\mathtt{AB}$,
    where $\revert{\mathtt{C}} \in hr(\headStr_m(P^\prime))$.
    If we have traversed with
    $\bodyStr_m(P^\prime)[1..1] = \mathtt{A}$
    from the source, we are at state $[\mathtt{A}]=[\mathtt{CA}]$,
    and $\mathtt{A}\neq\protect\longest{\mathtt{A}} = \mathtt{CA}$.
    At this point, it is guaranteed that all occurrences of 
    $\mathtt{A}$ (and all extensions to $\mathtt{A}$ that can be traversed on the {\PackedDAWG})
    will be immediately preceded by $\mathtt{C}$.
    Thus we only need to check outgoing edges.
    Since there is no outgoing edge labeled with the next
    meta-character $\mathtt{B}$,    
    $\tailStr_m(f^\prime)$ is the longest lcp between labels of 
    outgoing edges from $[\mathtt{CA}]$, which is the lcp
    between the successor $\mathtt{B}^+ = \mathtt{C}$ and $\mathtt{B}$
    (the predecessor $\mathtt{B}^-$ of $\mathtt{B}$ is $\mathbf{nil}$).
    }
  \label{fig:points}
\end{figure}

\begin{figure}[ph]
\centerline{
\includegraphics[width=1.0\textwidth]{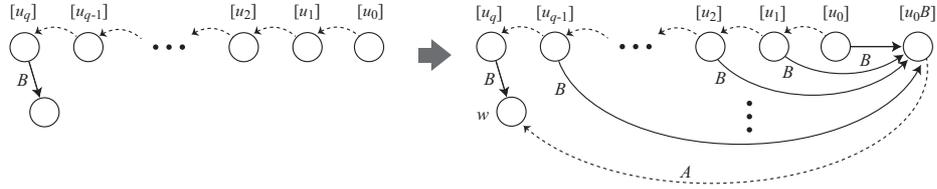}
}
\caption{Illustration for case (a) of Lemma~\ref{lem:Points}.
Primary edge $([u_0], B, [u_0B])$ and 
secondary edges $([u_j], B, [u_0B])$ are created for 
all $1 \leq j < q$.
The suffix link of the new sink $[u_0B]$ points to state $w$.
}
\label{fig:secondary_edges}
\end{figure}

\begin{figure}[ph]
\centerline{
\includegraphics[width=0.9\textwidth]{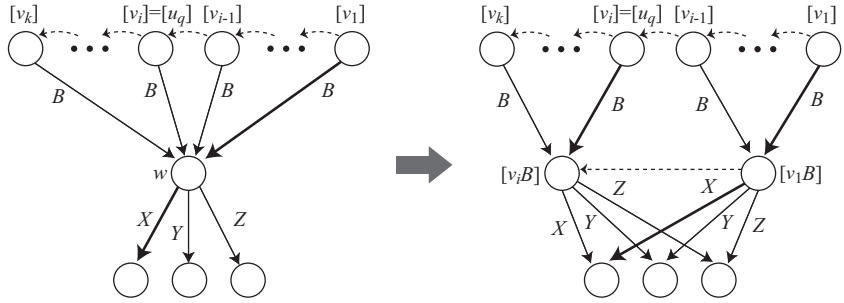}
}
\caption{
Illustration for case (b) of Lemma~\ref{lem:Points}.
State $w$ is split into two states $[v_1B]$ and $[v_kB]$
due to the update of the {\PackedDAWG}.
The bold arcs represent primary edges.
The outgoing edges of $w$ may or may not be primary,
however, the copied outgoing edges from $[v_iB]$ are all secondary.
}
\label{fig:state_separation}
\end{figure}

\begin{figure}[ph]
  \centerline{
   \includegraphics{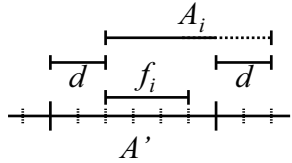}
  }
  \caption{
    Finding a meta-block $A^\prime$ that has
    $A_i[1..r-d]$ as a suffix (Lemma~\ref{lemma:getmetablock}).
    We can find $A^\prime$ and $d \leq d_{A,k}$ from $A_i$ in $O(\log N)$ time.
    Notice that $A_i$ does not necessarily occur at the position depicted, 
    but if we assume that there exists an occurrence of $A_i$ such that
    its prefix $f_i$ is inside a meta-block,
    $f_i$ is guaranteed to be completely contained in $A^\prime$
    since $A^\prime$ overlaps with $A_i$
    at least as much as any meta-block actually occurring
    prior to $A_i$.
    Thus, a previous occurrence of $f_i$ can be retrieved from
    $\minpos_{[A^\prime]}$, $d$, and $|f_i|$.
  }
  \label{fig:insideMetaBlock}
\end{figure}

\clearpage
\section*{Appendix B: Pseudo Codes}
\begin{algorithm2e}
  \SetKw{KwAnd}{and}
  \caption{On-line DAWG construction algorithm by Blumer et al.~\cite{blumer85:_small_autom_recog_subwor_text}}
  \label{algorithm:onlineDAWGconstruction}
  \underline{\textbf{Procedure} \textsf{builddawg}(\textsf{w})}\;
  Create state \textsf{source}; \textsf{cursink} = \textsf{source}\;
  \lFor{each letter \textsf{a} of \textsf{w}}{
    \textsf{cursink} = \textsf{update}(\textsf{cursink}, \textsf{a})\;
  }
  \Return \textsf{source}\;
  \vspace{0.2cm}
  \underline{\textbf{Procedure} \textsf{update(\textsf{cursink}, \textsf{a})}}\;
  Create state \textsf{newsink} and primary edge
  (\textsf{cursink}, \textsf{a}, \textsf{newsink})\;
  \textsf{curstate} = \textsf{cursink}; \textsf{suffixstate} = undefined\;
  \While{\textsf{curstate} $\neq$ \textsf{source} \KwAnd \textsf{suffixstate} = undefined}{
    \textsf{curstate} = \textit{suflink}(\textsf{curstate})\;
    \uIf{$\not\exists\textsf{cstate}$
      s.t. $(\textsf{curstate},\textsf{a},\textsf{cstate})\in E$}{
    create secondary edge $(\textsf{curstate},\textsf{a},\textsf{newsink})$\;
    }\uElseIf{$(\textsf{curstate},
      \textsf{a},\textsf{cstate})\in E$
      is primary}{
      \textsf{suffixstate} = \textsf{cstate}\;
    }\Else(\tcc*[h]{$(\textsf{curstate},
      \textsf{a},\textsf{cstate})\in E$
      is secondary}){
      \textsf{suffixstate} =
      $\textsf{split}(\textsf{curstate},\textsf{a},\textsf{cstate})$\;
    }
  }
  \vspace{0.2cm}
  \underline{\textbf{Procedure} \textsf{split}(\textsf{pstate},
    \textsf{a}, \textsf{cstate})}\;
  Create state \textsf{newcstate}\;
  Change secondary edge $(\textsf{pstate},\textsf{a},\textsf{cstate})$
  to primary edge $(\textsf{pstate},\textsf{a},\textsf{newcstate})$\;
  \For{every edge $(\textsf{cstate}, \textsf{c}, \textsf{dest})\in E$}{
    create secondary edge $(\textsf{newcstate}, \textsf{c},\textsf{dest})$\;
  }
  \textit{suflink}(\textsf{newcstate}) = \textit{suflink}(\textsf{cstate})\;
  \textit{suflink}(\textsf{cstate}) = \textsf{newcstate}\;
  \textsf{curstate} = \textsf{pstate}\;
  \While{\textsf{curstate} $\neq$ \textsf{source}}{
    \textsf{curstate} = \textit{suflink}(\textsf{curstate})\;
    \uIf{$\exists\textsf{cstate}$ s.t. $(\textsf{curstate},\textsf{a},\textsf{cstate})\in E$}{
      Change secondary edge
      $(\textsf{curstate},\textsf{a},\textsf{cstate})$ to
      secondary edge
      $(\textsf{curstate},\textsf{a},\textsf{newcstate})$\;
    }\Else{
      break out of the while loop\;
    }
  }
  \Return \textsf{newcstate}\; 
\end{algorithm2e}

\clearpage
\section*{Appendix C: Proof of Theorem~\ref{theo:RLE}}

Here we provide a proof for Theorem~\ref{theo:RLE} and
show how to compute the s-factorization of a string $S$ 
from $\RLE(S)$, efficiently and on-line.
We begin with the following lemma.

\begin{lemma} \label{lem:RLE_LZ_linear}
  Each RL factor $a_k^{p_k}$ of $\RLE(S)$ is covered by 
  at most 2 s-factors of string $S$.
\end{lemma}
\begin{proof}
  Consider an s-factor $f_i$ that starts at the $j$th position
  in the RL factor $a_k^{p_k}$, where $1 < j \leq p_k$.
  Since $a_k^{p_k - j + 1}$ is both a suffix and a prefix of $a_k^{p_k}$,
  we have that the s-factor extends at least to the end of $a_k^{p_k}$.
  This implies that each RL factor $a_k^{p_k}$ is always covered by
  at most 2 s-factors.
\qed
\end{proof}
Let $z$ be the number of s-factors of string $S$.
It immediately follows from Lemma~\ref{lem:RLE_LZ_linear} that $z \leq 2m$.
This allows us to describe the complexity of our algorithm
without using $z$.
Lemma~\ref{lem:RLE_LZ_linear} also implies that 
if an s-factor $f_i$ intersects with an RL factor $a_{k}^{p_k}$,
then the first RL factor of $f_i$
is always a suffix $a_{k}^p$ of $a_{k}^{p_k}$ with $p \leq p_k$.
This simplicity allows us to perform on-line s-factorization from RLE efficiently.
A proof for Theorem~\ref{theo:RLE} follows:

\begin{proof}
Let $\RLE(S) = a_1^{p_1} a_2^{p_2} \cdots a_{m}^{p_m}$.
For any $1 \leq k \leq h \leq m$,
let $\RLE(S)[k..h] = a_k^{p_k} a_{k+1}^{p_{k+1}} \cdots a_h^{p_h}$.
Let $\Substr(\RLE(S)) = \{\RLE[k..h] \mid 1 \leq k \leq h \leq m\}$.

Assume we have already computed $f_1, \ldots, f_{i-1}$
and we are computing a new s-factor $f_i$ 
from the $(\ell_i+1)$th position of $S$.
Let $a^d$ be the RL factor which contains the $(\ell_i+1)$th position,
and let $j$ be the position in the RL factor where $f_i$ begins.

Firstly, consider the case where $2 \leq j \leq d$.
Let $p = d - j + 1$, i.e., the remaining suffix of $a^d$ is $a^p$.
It follows from Lemma~\ref{lem:RLE_LZ_linear} that 
$a^p$ is a prefix of $f_i$.
In the sequel, we show how to compute the rest of $f_i$.
For any out-going edge $e = ([u], b^q, [ub^q])$ of a state $[u]$ of the RLE-DAWG for $\RLE(S)[1..j]$
and each character $a \in \Sigma$,
define
\[
 \maxe_{[u]}(a,b^q) = \max(\{ p \mid a^p \longest{u} b^q \in \Substr(\RLE(S)[1..j])\}\cup\{0\}).
\]
That is, $\maxe_{[u]}(a,b^q)$ represents the maximum exponent of the
RL factor with character $a$, that immediately precedes $\longest{u}b^q$ in $\RLE(S)[1..j]$.
For each pair $(a, b)$ of characters for which there is an out-going edge 
$([u], b^q, [ub^q])$ from state $[u]$ and $\maxe_{[u]}(a,b^q) > 0$,
we insert a point $(\maxe_{[u]}(a,b^q), q)$ into
$\RLEPoints_{[u],a,b}$.
By similar arguments to the case of {\PackedDAWG}s,
each point in
$\RLEPoints_{[u],a,b}$ corresponds to a secondary edge, or a suffix
link (labeled with $a^p$ for some $p$) of a primary child, 
so the total number of such points is bounded by $O(m)$.

Suppose we have successfully traversed the RLE-DAWG by 
$u$ with an occurrence that is immediately preceded by $a^p$ 
(i.e., $a^pu$ is a prefix of s-factor $f_i$),
and we want to traverse with the next RLE factor $b^q$ from state $[u]$.

If $u = \longest{u}$, i.e., only primary edges were traversed,
then
we query $\RLEPoints_{[u],a,b}$
for a point with maximum $x$-coordinate in the
range $[0, N] \times [q, N]$.
Let $(x, y)$ be such a point.
If $x \geq p$, then since $y \geq q$,
there must be a previous occurrence of $a^p \longest{u} b^q$,
and hence $a^p \longest{u} b^q$ is a prefix of $f_i$.
If there is an outgoing edge of $[u]$ labeled by $b^q$,
then we traverse from $[u]$ to $[ub^q]$ and update the RLE-DAGW with the next RL factor.
Otherwise, it turns out that $f_i = a^p \longest{u} b^q$.
If $x < p$, or no such point existed, 
then we query for a point with maximum $y$-coordinate 
in the range $[p, N] \times [0, q]$.
If $(x^\prime,y^\prime)$ is a such a point,
then $f_i = a^p \longest{u} b^{y^\prime}$. 

Otherwise (if $u \neq \longest{u}$),
then all occurrences of $u$ in $S[1..\ell_i]$ is immediately preceded by 
the unique RL factor $a^p$.
Thus, there exists an occurrence of $a^p u b^q$ iff $([u], b^q, [ub^q]) \in E$.
If there is no such edge,
then the last RL factor of $f_i$ is $b^y$,
where $y = \min(\max\{k \mid ([u], b^k, [ub^k]) \in E\} \cup  \{q\})$.

Secondly, let us consider the case where $j = 1$.
Let $([\varepsilon], a^g, [a^g])$ be the edge which has 
maximum exponent $g$ for the character $a$ from the source state $[\varepsilon]$.
If $g < d$, then $f_i = a^g$.
Otherwise, $a^d$ is a prefix of $f_i$,
and we traverse the RLE-DAWG in a similar way as above,
while checking an immediately preceding occurrence of $a^d$.

If we use priority search trees by McCreight (\textit{SIAM J. Comput.} 14(2),
257--276, 1985) %
and balanced binary search trees,
the above queries and updates are supported in $O(\log m)$ time 
using a total of $O(m \log N)$ bits of space.
We can do better based on the following observation.
For a set $T$ of points in a 2D plane,
a point $(p, q) \in T$ is said to be dominant 
if there is no point $(p^\prime, q^\prime) \in T$
satisfying both $p^\prime \geq p$ and $q^\prime \geq q$.
Let $\RLEDomPoints_{[u],a,b}$ denote the set of dominant points of $\RLEPoints_{[u],a,b}$.
Now, a query for a point with maximum $x$-coordinate 
in range $[0, N] \times [q, N]$ reduces to
a successor query on the $y$-coordinates of points in $\RLEDomPoints_{[u],a,b}$
(see also Fig.~\ref{fig:successor1}).
On the other hand, a query for a point with maximum $y$-coordinate 
in range $[p, N] \times [0, q]$ reduces to
a successor query on the $x$-coordinate of points in $\RLEDomPoints_{[u],a,b}$
(see also Fig.~\ref{fig:successor2}).
Hence, it suffices to maintain only the dominant points.

\begin{figure}[t]
  \centerline{
    \includegraphics[scale=0.45]{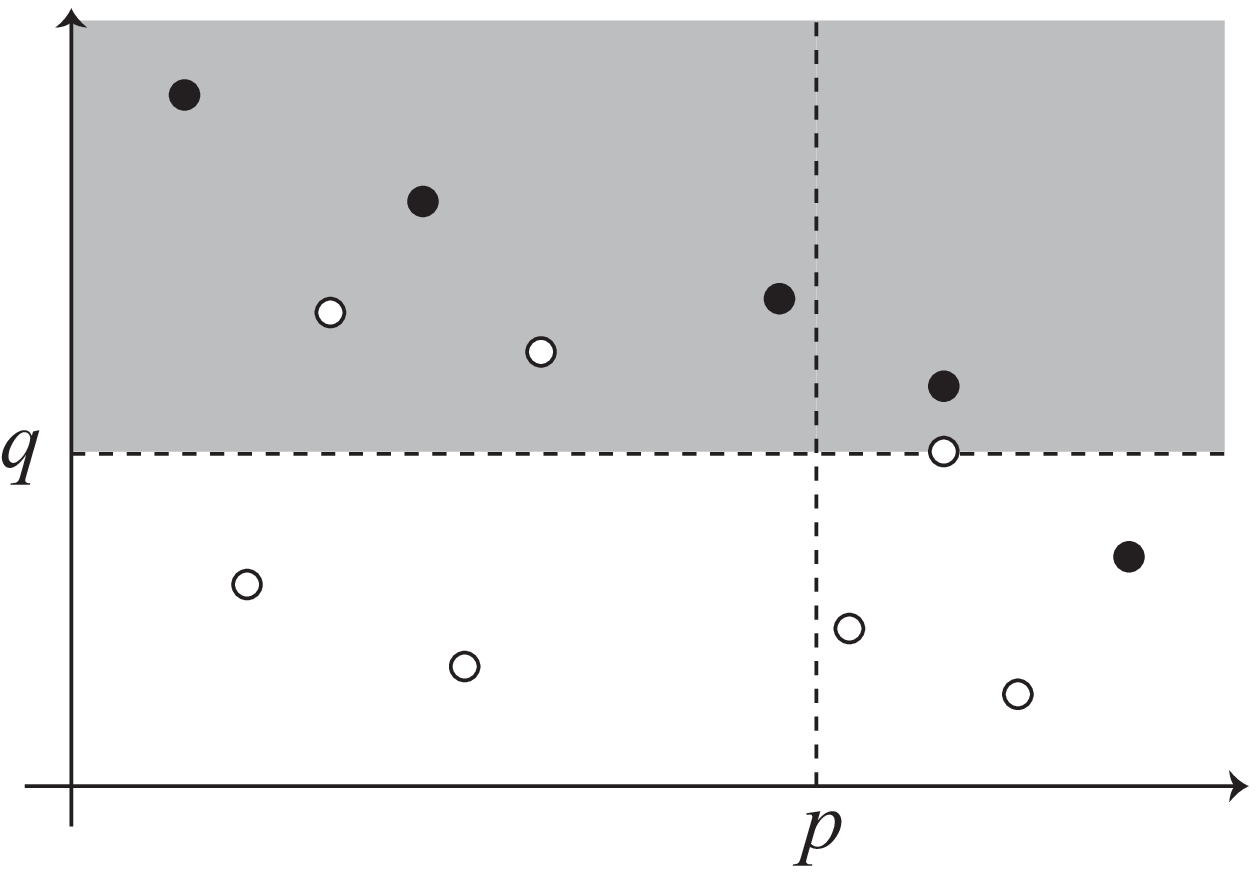}
    \hfill
    \includegraphics[scale=0.45]{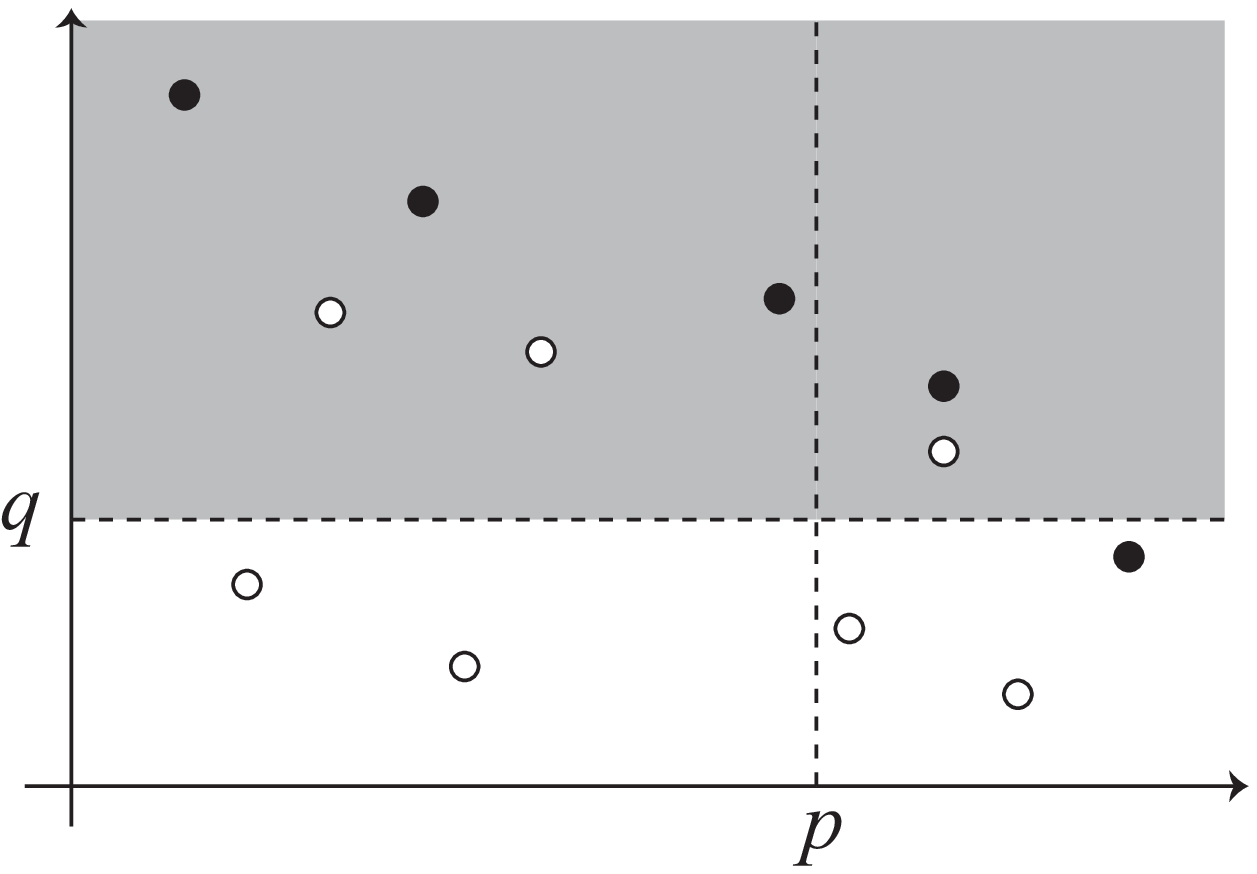}
  }
  \caption{    
  Black and white points are dominant and non-dominant points, respectively.
  Left: If there is a point $(x, q)$ with $x \geq p$,
  then we can traverse from state $[u]$ to $[ub^q]$ and continue.  
  Right: If there is no such point,
  then the last RL factor of $f_i$ is $b^q$.
  In both cases, the $y$-coordinate successor of $q$ among the dominant points
  is a point with maximum $x$-coordinate in range $[0, N] \times [q, N]$.
  }
\label{fig:successor1}
\end{figure}

\begin{figure}[t]
  \centerline{
    \includegraphics[scale=0.45]{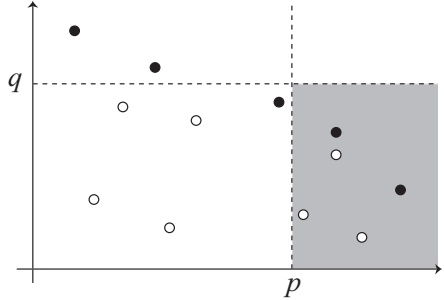}
  }
  \caption{    
  Black and white points are dominant and non-dominant points, respectively.
  Assume the $x$-coordinate of the point with maximum $x$-coordinate 
  in range $[0, N] \times [q, N]$ is less than $p$. 
  Then the $x$-coordinate successor of $p$ among the dominant points
  is a point with maximum $y$-coordinate in range $[p, N] \times [0, q]$.
  }
\label{fig:successor2}
\end{figure}

When a new dominant point is inserted into $\RLEDomPoints_{[u],a,b}$
due to an update of the RLE-DAWG,
then all the points that have become non-dominant are deleted from $\RLEDomPoints_{[u],a,b}$.
We can find each non-dominant point by a single predecessor/successor query.
Once a point is deleted from $\RLEDomPoints_{[u],a,b}$, it will never be re-inserted to $\RLEDomPoints_{[u],a,b}$.
Hence, the total number of insert/delete operations is linear in the size of $\RLEDomPoints_{[u],a,b}$,
which is $O(m)$ for all the states of the RLE-DAWG.
Using the data structure of~\cite{BeameF02},
predecessor/successor queries and insert/delete operations are 
supported in $O\left(\min \left\{\frac{(\log\log m)(\log \log N)}{\log\log\log N}, \sqrt{ \frac{\log m}{\log \log m}} \right\}\right)$ time,
using a total of $O(m \log N)$ bits of space.

Each state of the RLE-DAWG has at most $m$ children and 
the exponents of the edge labels are in range $[1, N]$.
Hence, assuming an integer alphabet $\Sigma = \{1, 2, \ldots, N\}$ and
using the data structure of~\cite{BeameF02}, we
can search branches at each state in $O\left(\min \left\{\frac{(\log\log m)(\log \log N)}{\log\log\log N}, \sqrt{ \frac{\log m}{\log \log m}} \right\}\right)$ time,
using a total of $O(m \log N)$ bits of space.
A final technicality is how to access
 the set $\RLEDomPoints_{[u],a,b}$ which 
is associated with a pair $(a, b)$ of characters.
To access $\RLEDomPoints_{[u],a,b}$ at each state $[u]$,
we maintain two level search structures,
one for the first characters and the other for the second characters
of the pairs.
At each state $[u]$
we can access $\RLEDomPoints_{[u],a,b}$ in $O\left(\min \left\{\frac{(\log\log m)(\log \log N)}{\log\log\log N}, \sqrt{ \frac{\log m}{\log \log m}} \right\}\right)$ time
with a total of $O(m \log N)$ bits of space, again using
the data structure of~\cite{BeameF02}.
This completes the proof for Theorem~\ref{theo:RLE}.
\qed
\end{proof}